\documentclass[10pt,a4paper,oneside]{amsart}
\usepackage[a4paper]{geometry}

\usepackage[english]{babel}  
\usepackage{latexsym, amsthm}
\usepackage{amsfonts, amsmath, amssymb}
\usepackage{url}
\usepackage{booktabs}
\usepackage{graphicx}
\usepackage[usenames]{color}
\usepackage[all]{xy}
\usepackage{graphicx}

\usepackage{algorithm}
\usepackage[noend]{algpseudocode}

\newcommand{\?}[1]{
  \sbox0{A#1}\sbox2{A\kern0pt #1}%
  \kern\dimexpr\wd0-\wd2\relax
  #1%
}

\algrenewcommand\alglinenumber[1]{{\scriptsize#1}}   
\algrenewcommand\algorithmicrequire{\textbf{Input:}} 
\algrenewcommand\algorithmicensure{\textbf{Output:}} 
\algnewcommand\algorithmicassume{\textbf{Assumption:}} 
\algrenewcommand\algorithmicfunction{\textbf{Fun}}
\algnewcommand\Assume{\item[\algorithmicassume]}

\newcommand{\assign}{:=}

\newcommand{\inlcomment}[1]{\texttt{\small/* #1 */}}
\newcommand{\eolcomment}[1]{\hfill\texttt{\small// #1}}

\usepackage[capitalize]{cleveref}

\theoremstyle{plain}
\newtheorem{theorem}{Theorem}[section]

\newtheorem{proposition}[theorem]{Proposition}
\newtheorem{lemma}[theorem]{Lemma}
\newtheorem{corollary}[theorem]{Corollary}
\newtheorem{definition}[theorem]{Definition}
\theoremstyle{remark}
\newtheorem{XxmpX}[theorem]{Example}
\newtheorem{XrmkX}[theorem]{Remark}

\newenvironment{example}
  {%
   \pushQED{\qed}\begin{XxmpX}}
  {\popQED\end{XxmpX}}
\newenvironment{remark}
  {%
   \pushQED{\qed}\begin{XrmkX}}
  {\popQED\end{XrmkX}}

\definecolor{orange}{rgb}{1,0.5,0}
\definecolor{purple}{rgb}{1,0,0.5}
\definecolor{darkgreen}{rgb}{0,0.5,0}

\newcommand{\Z}{\mathbb{Z}}
\newcommand{\N}{\mathbb{Z}_{\ge 0}}

\def\F{{\mathbb F}}

\begin{document}

\title{Two-Point Codes for the Generalised GK curve}

\author{\'E{}lise Barelli, Peter Beelen, Mrinmoy Datta, Vincent Neiger, Johan Rosenkilde}
\thanks{\'Elise Barelli is partially supported by a DGA-MRIS scholarship and a
French ANR-15-CE39-0013-01 ``Manta''. Peter Beelen gratefully acknowledges
the support by The Danish Council for Independent Research (Grant
No.\,DFF--4002-00367). Vincent Neiger has received funding from the People
Programme (Marie Curie Actions) of the European Union's Seventh Framework
Programme (FP7/2007-2013) under REA grant agreement number 609405
(COFUNDPostdocDTU). Mrinmoy Datta is supported by The Danish Council for
Independent Research (Grant No.\,DFF6108-00362).}%
\thanks{\'E{}lise~Barelli is with INRIA Saclay and LIX, \'Ecole Polytechnique,
91120 Palaiseau Cedex, France (e-mail: elise.barelli@inria.fr).
Peter~Beelen, Mrinmoy~Datta, Vincent~Neiger, and Johan~Rosenkilde are
with the Department of Applied Mathematics and Computer Science, Technical
University of Denmark, 2800 Kgs. Lyngby, Denmark (e-mails: pabe@dtu.dk,
mrinmoy.dat@gmail.com, jsrn@dtu.dk).
Vincent~Neiger is also with XLIM, Universit\'e de Limoges, 87060
Limoges Cedex, France (e-mail: vincent.neiger@unilim.fr).}

\maketitle

\begin{abstract}
We improve previously known lower bounds for the minimum distance of certain two-point AG codes constructed using a Generalized Giulietti--\?Korchmaros curve (GGK).
Castellanos and Tizziotti recently described such bounds for two-point codes coming from the Giulietti--\?Korchmaros curve (GK).
Our results completely cover and in many cases improve on their results, using different techniques, while also supporting any GGK curve.
Our method builds on the order bound for AG codes: to enable this, we study certain Weierstrass semigroups.
This allows an efficient algorithm for computing our improved bounds.
We find several new improvements upon the MinT minimum distance tables.
\end{abstract}

\section{Introduction}

Algebraic geometry (AG) codes are a class of linear codes constructed from algebraic curves defined over a finite field. This class continues to provide examples of good codes when considering their basic parameters: the length $n$, the dimension $k$, and the minimum distance $d$. If the algebraic curve used to construct the code has genus $g$, the minimum distance $d$ satisfies the inequality $d \ge n-k+1-g$. This bound, a consequence of the Goppa bound, implies that the minimum distance of an AG code can be designed. It is well known that the Goppa bound is not necessarily tight, and there are various results and techniques which can be used to improve upon it in specific cases. Such a result has been given in \cite[Thm.\,2.1]{M}, where the Goppa bound is improved by one. Another approach to give lower bounds on the minimum distance of AG codes is described in \cite{handbook} and the references therein. This type of lower bound is often called the \emph{order bound}; various refinements and generalizations have been given, for example in \cite{B,DKP}.

To obtain good AG codes, the choice of the algebraic curve in the construction plays a key role. A very good class of curves are the so-called maximal curves, i.e., algebraic curves defined over a finite field having as many rational points as allowed by the Hasse--\?Weil bound. More precisely, a maximal curve of genus $g$ defined over a finite field $\F_q$ with $q$ elements, has $q+1+2\sqrt{q}g$ $\F_q$-rational points, i.e., points defined over $\F_q$; this only makes sense if the cardinality $q$ is a square number. An important example of a maximal curve is the Hermitian curve, but recently other maximal curves have been described \cite{GK,GGS}, often called the \emph{generalized Giulietti--\?Korchm\'aros (GK) curves}. In this article we continue the study of two-point AG codes coming from the generalized GK curves that was initiated in \cite{CT}. However, rather than using the improvement upon the Goppa bound from \cite{M}, we use the order bound as given in \cite{B}. As a matter of fact, we also show that the order bound from \cite{B} implies Theorem 2.1 in \cite{M}. Thus, we will automatically recover all the results in \cite{CT}, but on various occasions we obtain better bounds for the minimum distance than the ones reported in \cite{CT}. We will also paraphrase the order bound from \cite{B} and explain how we have computed it. A key object in this computation is a two-point generalization of a Weierstrass semigroup given in \cite{BT}, and therefore some time will be used to describe this semigroup explicitly in the case of certain pairs of points on the generalized GK curve.

After finishing this work, we were made aware of the contemporaneous work \cite{HY}. In \cite{HY} multi-point codes and their duals from the generalized GK function field are constructed and investigated. Proposition \ref{prop:dual} is different from, but akin to \cite[Thm.\,2]{HY} and similar proof techniques were used. The techniques used in \cite{HY} to analyse the code parameters are very different from ours and more related to the ones used in \cite{CT}. Our main tools, the explicit computation of the map $\tau_{0,\infty}$ in Corollary \ref{cor:tau} and the resulting algorithm to compute the order bound, were not employed in \cite{HY}. Our improvements on the MinT code tables are not present in \cite{HY}.

\section{Preliminaries}

Though later we will only consider the generalized GK curves, we will in this section consider any algebraic curve $\chi$ defined over a finite field $\F_q$. The field of functions on $\chi$, or briefly the function field of $\chi$, will be denoted by $\F_q(\chi)$, while the genus of $\chi$ is denoted by $g(\chi)$. Rather than using the language of curves, we will formulate the theory using the language of function fields; see \cite{S} for more details. In particular, we will speak about places of $\F_q(\chi)$ rather than points of $\chi$. For any place $Q$ of $\F_q(\chi)$, we denote by $v_Q$ the valuation map at the place $Q$. The valuation $v_Q: \F_q(\chi)\setminus \{0\} \rightarrow \Z$ sends a nonzero function $f$ to its order of vanishing at $Q$. If $v_Q(f)<0$, one also says that $f$ has a pole of order $-v_Q(f)$ at $Q$.

A divisor of $\F_q(\chi)$ is a finite formal sum $\sum_i n_i Q_i$ of places $Q_i$ of $\F_q(\chi)$, where the $n_i$'s are integers in $\Z$. The support of a divisor $\sum_i n_i Q_i$ is the (finite) set of places $\{Q_i \mid n_i \neq 0\}.$ Finally, we call two divisors disjoint if they have disjoint supports. To any nonzero function $f \in \F_q(\chi)$ one can associate two divisors $(f)$ and $(f)_\infty$ known as the divisor of $f$ and the divisor of poles of $f$ respectively, given by:
$$(f):=\sum_Q v_Q(f)\, Q \quad \makebox{and} \quad (f)_\infty:=\sum_{Q;v_Q(f)<0}-v_Q(f)\, Q.$$
If all the coefficients $n_i$ in a divisor $G=\sum_i n_i Q_i$ are nonnegative, we call $G$ an effective divisor; notation $G \ge 0$.

We now recall some notations for AG codes; we again refer to \cite{S} for a more comprehensive exposition. Let $ \mathcal{P} = \{ P_1,\dots,P_n \}$ be a set of $n$ distinct rational places of $\F_q(\chi)$, i.e., places of degree $1$, and define the divisor $D=P_1+\cdots+P_n$. Further let $G$ be a divisor
such that $\deg(G) < n$ and $G$ does not contain any place of $\mathcal{P}$. We consider the following map:
\[ \begin{array}{cccl}
    Ev_{\mathcal{P}}: & \F_q(\chi)_{\mathcal{P}} & \longrightarrow & \F^n \\
            & f & \longmapsto & (f(P_1),\dots,f(P_n)).
   \end{array} \]
Here $\F_q(\chi)_{\mathcal{P}}$ denotes the subset of $\F_q(\chi)$ consisting of functions not having a pole at any $P \in \mathcal{P}$. Then we define the AG code $C_L(D,G)$ by $C_L(D,G) := \{ Ev_{\mathcal{P}}(f) \mid f \in L(G) \}$.
Here $L(G)$ denotes the Riemann--\?Roch space $L(G):=\{f \in \F_q(\chi)\setminus\{0\} \mid (f)+G \ge 0\} \cup \{0\}$.
It is well known that the minimum distance $d$ of $C_L(D,G)$ (resp. $C_L(D,G)^\perp$) satisfies the Goppa bound $d \ge n - \deg(G)$ (resp. $d \ge \deg(G) -2g(\chi)+2$).

Here, we will make use of another lower bound for the minimum distance of
$C_L(D,G)^\perp$, obtained in \cite{B}. We will use the notions of $G$-gaps and $G$-non-gaps at a place $Q$, which were for example also used in \cite{GKL}.

\begin{definition}
Let $Q$ be a rational place and $G$ be a rational divisor of $\F_q(\chi)$.
We define $L(G+\infty Q) := \bigcup_{i \in \Z}{L(G + iQ)}$ and
\[H(Q;G) := \{-v_Q(f) \mid f \in L(G+\infty Q) \setminus \{0\} \}. \]
We call $H(Q;G)$ the set of \emph{$G$-non-gaps at $Q$}. The set
\[ \Gamma(Q;G):=\Z_{\ge v_Q(G)-\deg(G)} \setminus H(Q;G) \]
is called the set of \emph{$G$-gaps at $Q$}.
\end{definition}

Note that if $G=0$ we obtain $H(Q;0)=H(Q)$, the Weierstrass semigroup of $Q$, and $\Gamma(Q;0)=\Gamma(Q)$, the set of gaps at $Q$. Further, note that if $i \in H(Q;F_1)$ and $j \in H(Q;F_2)$, then $i+j \in H(Q;F_1+F_2)$. Finally, observe that the theorem of Riemann--\?Roch implies that the number of $G$-gaps at $Q$ coincides with the genus of $\chi$, that is, $|\Gamma(Q;G)|=g(\chi)$.

\begin{remark}\label{rem:2.2}
If $i < -\deg(G)$ then $\deg(G+iQ) < 0$ and $L(G+iQ) = \{0\}$. So in the previous
definition we can write $L(G+\infty Q) = \bigcup_{i \ge -\deg(G)}{L(G + iQ)}$.
Further, note that for any $a\in \Z$ we have $L(G+aQ+\infty Q)=L(G+\infty Q)$ and hence $H(Q;G+aQ)=H(Q;G)$ as well as $\Gamma(Q;G+aQ)=\Gamma(Q;G)$.
\end{remark}

\begin{definition}
Let $Q$ be a rational place and let $F_1$, $F_2$ be two divisors of $\chi$. As in \cite{B} we define
\begin{align*}
N(Q;F_1,F_2) &:=\{(i,j) \in H(Q;F_1) \times H(Q;F_2) \mid i+j = v_Q(G) + 1 \}, \\
\nu(Q;F_1,F_2) &:= |N(Q;F_1,F_2)|.
\end{align*}
\end{definition}

\begin{proposition}\cite[Prop.\,4]{B}\label{prop:B}
Let $D=P_1+\cdots + P_n$ be a divisor that is a sum of $n$ distinct rational places of $\F_q(\chi)$, $Q$ be
a rational place not occurring in $D$, and $F_1,F_2$ be two divisors disjoint from $D$. Suppose that $C_L(D,F_1+F_2) \neq C_L(D,F_1+F_2 + Q)$. Then, for any codeword $c \in C_L(D,F_1+F_2)^{\perp}  \setminus C_L(D,F_1+F_2 + Q)^{\perp}$, we have $$w_H(c) \ge \nu(Q;F_1,F_2).$$ In particular, the minimum distance $d(F_1+F_2)$ of $C_L(D,F_1+F_2)^\perp$ satisfies
$$d(F_1+F_2) \ge \min\{\nu(Q;F_1,F_2),d(F_1+F_2+Q)\},$$ where $d(F_1+F_2+Q)$ denotes the minimum distance of $C_L(D,F_1+F_2+Q)^\perp.$
\end{proposition}

To arrive at a lower bound for the minimum distance of $C_L(D,G)^{\perp}$, one applies this proposition in a recursive manner. More precisely, one constructs a sequence $Q^{(1)}, \dots, Q^{(N)}$ of not necessarily distinct rational places, none occurring in $D$, such that $C_L(D,G+Q^{(1)}+\cdots+Q^{(N)})^{\perp}=0$.
Such a sequence exists, since the theorem of Riemann--\?Roch implies that $C_L(D,G+Q^{(1)}+\cdots+Q^{(N)})=\F_q^n$ as soon as $N \ge 2g(\chi)-1+n-\deg(G)$. Then, the code $C_L(D,G)^\perp$ has for example minimum distance at least $\min \nu(Q^{(i)};G+Q^{(1)}+\cdots+Q^{(i-1)},0)$, where the minimum is taken over all $i$ satisfying $1 \le i \le N$ and $C_L(D,G+Q^{(1)}+\cdots+Q^{(i-1)}) \neq C_L(D,G +Q^{(1)}+\cdots+Q^{(i)}).$

The well known Goppa bound is a direct consequence of \cref{prop:B} as shown in \cite[Lem.\,9]{B}. We will need the following slightly more general version of \cite[Lem.\,9]{B}.
\begin{lemma}\label{lem:B}
Let $D=P_1+\cdots + P_n$ be a sum of distinct rational places, let $Q$ be a rational place not occurring in $D$, and let $F_1,F_2$ be two divisors disjoint from $D$. Then $\nu(Q;F_1,F_2) \ge \deg(F_1+F_2)-2g+2.$
\end{lemma}
\begin{proof}
Define the formal Laurent series $$p_{Q;F_1}(t):=\sum_{i \in H(Q;F_1)} t^i \hspace{1cm} \makebox{and} \hspace{1cm} p_{Q;F_2}(t):=\sum_{i \in H(Q;F_2)} t^i.$$
Then $\nu(Q;F_1,F_2)$ is the coefficient of $t^{v_Q(F_1+F_2)+1}$ in the Laurent series $p_{Q;F_1}(t)\cdot p_{Q;F_2}(t)$. The lemma follows by analyzing this product carefully. First we introduce
$$q_{Q;F_1}(t):=\sum_{i \in \Gamma(Q;F_1)} t^i \hspace{1cm} \makebox{and} \hspace{1cm} q_{Q;F_2}(t):=\sum_{i \in \Gamma(Q;F_2)} t^i.$$
Then $$p_{Q;F_1}(t)+q_{Q;F_1}(t)=\dfrac{t^{v_Q(F_1)-\deg(F_1)}}{1-t} \hspace{1cm} \makebox{and} \hspace{1cm} p_{Q;F_2}(t)+q_{Q;F_2}(t)=\dfrac{t^{v_Q(F_2)-\deg(F_2)}}{1-t},$$
implying that
\begin{multline*}
p_{Q;F_1}(t)\cdot p_{Q;F_2}(t)=
t^{v_Q(F_1+F_2)-\deg(F_1+F_2)} \left( \dfrac{1}{(1-t)^2}-\frac{2g(\chi)}{1-t}\right.\\
\left.+\dfrac{g(\chi)-t^{-v_Q(F_2)+\deg(F_2)}q_{Q;F_2}(t)}{1-t}
+\dfrac{g(\chi)-t^{-v_Q(F_1)+\deg(F_1)}q_{Q;F_1}(t)}{1-t} \right)
+q_{Q;F_1}(t)\cdot q_{Q;F_2}(t).
\end{multline*}
Since both $t^{-v_Q(F_1)+\deg(F_1)}q_{Q;F_1}(t)$ and $t^{-v_Q(F_2)+\deg(F_2)}q_{Q;F_2}(t)$ are a sum of $g(\chi)$ distinct nonnegative powers of $t$, the last three Laurent series in the above expression are in fact finite Laurent series with nonnegative coefficients. Hence the coefficient of $t^{v_Q(F_1+F_2)+1}$ in $p_{Q;F_1}(t)\cdot p_{Q;F_2}(t)$ is bounded from below by the corresponding coefficient in $$t^{v_Q(F_1+F_2)-\deg(F_1+F_2)}\left(1/(1-t)^2-2g(\chi)/(1-t) \right),$$ which is $\deg(F_1+F_2)-2g(\chi)+2.$
\end{proof}

In this paper, we are interested in a lower bound on the minimum distance for two-point AG codes. We will typically apply \cref{prop:B} to the special setting where $F_1=0$ and $F_2=G= a_1Q_1 + a_2Q_2$, with $Q_1,Q_2$ two rational places of $\F_q(\chi)$. Hence we want to compute $\nu(Q;G):=\nu(Q;0,G)$ where $G = a_1Q_1 + a_2Q_2$. Furthermore, we will only consider the case where $Q \in \{Q_1,Q_2\}.$ In order to compute the number $\nu(Q;G)$, we need to know the Weierstrass semigroup
$H(Q)$ and the set $H(Q;G)$ of $G$-non-gaps at $Q$. A very practical object in this setting is a two-point generalization of the Weierstrass semigroup and a map between two Weierstrass semigroups considered in \cite{BT}:

\begin{definition}\label{def:tau}
Let $Q_1,Q_2$ be two distinct rational places of $\F_q(\chi)$. We define
$R(Q_1,Q_2) := \{ f \in \F_q(\chi) \mid \mathrm{Supp}((f)_\infty) \subseteq \{Q_1,Q_2\} \}$, the
ring of functions on $\chi$ that are regular outside the points $Q_1$ and $Q_2$.
The two-point Weierstrass semigroup of $Q_1$ and $Q_2$ is then defined as:
\[H(Q_1,Q_2) := \{ (n_1,n_2) \in \Z^2 \mid \exists f \in R(Q_1,Q_2)\setminus \{0\},~
v_{Q_i}(f) = -n_i, i \in \{1,2\}\}.\]
Further we define the following map:
\[ \begin{array}{cccl}
    \tau_{Q_1,Q_2}: & \Z & \longrightarrow & \Z \\
            & i & \longmapsto & \min\{j \mid (i,j) \in H(Q_1,Q_2)\}.
   \end{array} \]
\end{definition}

\begin{remark}\label{rem:tau}
Note that $H(Q_1,Q_2) \subseteq \{(i,j) \in \Z^2 \mid i+j \ge 0\}$, since $L(iQ_1+jQ_2)=\{0\}$ if $i+j < 0$. In particular, we have for any $i \in \Z$ that $\tau_{Q_1,Q_2}(i) \ge -i$. Moreover, the theorem of Riemann--\?Roch implies that $\tau_{Q_1,Q_2}(a_1) \le 2g(\chi)-a_1.$
\end{remark}

\begin{proposition}\cite[Prop.\,14]{BT} Let $Q_1,Q_2$ be two distinct rational places of
 $\F_q(\chi)$. The map $\tau_{Q_1,Q_2}$ is
bijective and $\tau_{Q_1,Q_2}^{-1} = \tau_{Q_2,Q_1}$.
\end{proposition}

By the definitions of $\tau_{Q_1,Q_2}$ and $H(Q_1,Q_2)$, for all $i \in \Z$ there exists a function
$f_{Q_1,Q_2}^{(i)} \in R(Q_1,Q_2)$ such that $v_{Q_1}(f_{Q_1,Q_2}^{(i)})=-i$ and
$v_{Q_2}(f_{Q_1,Q_2}^{(i)})=-\tau_{Q_1,Q_2}(i).$ Since $\tau_{Q_1,Q_2}$ is a bijection, the functions $f_{Q_1,Q_2}^{(i)}$ have distinct pole orders at $Q_1$ as well as $Q_2$.

\begin{theorem}\label{thm:dimL}
Let $Q_1,Q_2$ be two distinct rational places of $\chi$ and
$a_1,a_2 \in \N$.
The Riemann--\?Roch space $L(a_1Q_1 + a_2Q_2)$ has dimension
$|\{i \le a_1 \mid \tau_{Q_1,Q_2}(i) \le a_2\}|$ and basis
\[
  \{f_{Q_1,Q_2}^{(i)} \mid i \le a_1 \makebox{ and } \tau_{Q_1,Q_2}(i) \le a_2\}.
\]
\end{theorem}

\begin{proof}
Consider the filtration of $\F$-vector spaces:
\[L(a_1Q_1 + a_2Q_2) \supseteq L((a_1-1)Q_1 + a_2Q_2) \supseteq \dots \supseteq L(-a_2Q_1 + a_2Q_2) \supseteq L(-(a_2+1)Q_1+a_2Q_2)=\{0\}.\]
For $-a_2 \le i \le a_1$, the strict inequality $\ell(iQ_1 + a_2Q_2)>\ell((i-1)Q_1 + a_2Q_2)$ holds if and only if there exists a function $f \in \F_q(\chi)$ such that
$(f)_\infty = iQ_1 + jQ_2$ with $j \le a_2$. Such a function exists if and only if
$\tau_{Q_1,Q_2}(i) \le a_2$. Hence, $\ell(a_1Q_1 + a_2Q_2) = |\{ -a_2 \le i \le a_1 \mid \tau_{Q_1,Q_2}(i) \le a_2\}|$.
Since $\tau_{Q_1,Q_2}(i) \ge -i$, we see that $\ell(a_1Q_1 + a_2Q_2) = |\{i \le a_1 \mid \tau_{Q_1,Q_2}(i) \le a_2\}|$ as was claimed.

A basis for $L(a_1Q_1+a_2Q_2)$ can be directly derived from the above, since the set $$\{f_{Q_1,Q_2}^{(i)} \mid i \le a_1 \text{ and }\tau_{Q_1,Q_2}(i) \le a_2\}$$ is a subset of $L(a_1Q_1+a_2Q_2)$ consisting of $\ell(a_1Q_1 + a_2Q_2)$ linearly independent functions. Note that the linear independence follows from the fact the functions have mutually distinct pole orders at $Q_1$.
\end{proof}
A direct corollary is an explicit description of the $(a_1Q_1+a_2Q_2)$-gaps and non-gaps at $Q_1$.
\begin{corollary}\label{cor:Gnongaps}
Let $G=a_1Q_1+a_2Q_2$. Then the set of $G$-non-gaps at $Q_1$ is given by
\[\{ a \in \Z \mid \tau_{Q_1,Q_2}(a) \le a_2\}\]
and the set of $G$-non-gaps at $Q_2$ is given by
\[\{ b \in \Z \mid \tau_{Q_1,Q_2}^{-1}(b) \le a_1\}.\]
\end{corollary}
\begin{proof}
The first part follows directly from the previous theorem by considering basis of $L(aQ_1+a_2Q_2)$ for $a$ tending to infinity. Reversing the roles of $Q_1$ and $Q_2$, the second part follows.
\end{proof}
This corollary implies that for $G=a_1Q_1+a_2Q_2$, it is not hard to compute the $G$-gaps at either $Q_1$ or $Q_2$ once the bijection $\tau_{Q_1,Q_2}$ can be computed efficiently. We show in an example that this does occur in a particular case. Moreover, in the next section we will give a very explicit description of $\tau_{Q_1,Q_2}$ for a family of function fields and pairs of rational points $Q_1$ and $Q_2$.

\begin{example}
The Hermitian curve $\mathcal H$ is the curve defined over $\F_{q^2}$ by the equation $x^q+x=y^{q+1}$. The corresponding function field $\F_{q^2}(\mathcal H)$ is called the Hermitian function field. For any two distinct rational places $Q_1$ and $Q_2$ of $\F_{q^2}(\mathcal H)$, the map $\tau_{Q_1,Q_2}$ satisfies $\tau_{Q_1,Q_2}(i)=-iq$ for $q \le i \le 0$. Since furthermore $\tau_{Q_1,Q_2}(i+q+1)=\tau_{Q_1,Q_2}(i)-(q+1)$ for any $i \in \Z$, this describes $\tau_{Q_1,Q_2}$ completely. See \cite{BT} for more details. This example also appears as a special case in the next section.
\end{example}

\section{The generalized Giulietti--\?Korchm\'aros function field}
Let $e \ge 1$ be an odd integer. We consider the generalized Giulietti--\?Korchm\'aros (GK) curve $\chi_e$, also known as the Garcia--\?G\" uneri--\?Stichtenoth curve \cite{GGS}. It is defined over the finite field $\F_{q^{2e}}$ by the equations $$x^q+x=y^{q+1} \;\; \text{and} \;\; z^{\frac{q^e+1}{q+1}}=y^{q^2}-y.$$ This is a maximal curve when considered over the finite field $\F_{q^{2e}}$. Indeed, its genus and number of rational points are
\begin{align*}
  g(\chi_e) & := (q-1)(q^{e+1}+q^e-q^2)/2 , \\
  N_e       & := q^{2e+2}-q^{e+3}+q^{e+2}+1.
\end{align*}
As before, we will use the language of function fields and denote the corresponding function field $\F_{q^{2e}}(\chi_e)$ as the generalized GK function field. For $e=1$ one simply obtains the Hermitian function field $\F_{q^2}(\mathcal H)$, while for $e=3$, one obtains what is known as the Giulietti--\?Korchm\'aros function field \cite{GK}.

The function $x \in \F_{q^{2e}}(\chi_e)$ has exactly one zero and one pole, which we will denote by $Q_0$ and $Q_\infty$ respectively. The functions $y$ and $z$ also have a pole at $Q_\infty$ only. For a given rational place $P$ of $\F_{q^{2e}}(\chi_e)$ different from $Q_\infty$, we call $(x(P),y(P),z(P)) \in \F_{q^{2e}}^3$ the coordinates of $P$. For the function field $\F_{q^{2e}}(\chi_e)$, rational places are uniquely determined by their coordinates. A place with coordinates $(a,b,c) \in \F_{q^{2e}}^3$ will be denoted by $P_{(a,b,c)}$. In particular, we have $Q_0=P_{(0,0,0)}$.

With these notations, we can express the divisors of $x,y$ and $z$ as follows:
\begin{align*}
  (x) &= (q^e+1)(Q_0-Q_\infty), & \\
  (y) &= \sum_{\substack{a\in\F_q \\ a^{q}+a=0}}\frac{q^e+1}{q+1}P_{(a,0,0)} \ - \ q\frac{q^e+1}{q+1}Q_\infty, & \\
  (z) &= \sum_{\substack{(a,b) \in \F_{q^2} \\ a^q+a=b^{q+1}}}P_{(a,b,0)} \ - \ q^3Q_\infty.
\end{align*}
In each summation, the point $P_{(0,0,0)} = Q_0$ occurs. For future reference we also note that for $k \in \Z$ and $\ell\ge 0, m\ge 0$ we have
\begin{equation}\label{eq:divxyz}
(x^k y^\ell z^m)= \left(k(q^e+1)+\ell\frac{q^e+1}{q+1}+m\right)Q_0-\left( k(q^e+1)+\ell q\frac{q^e+1}{q+1}+mq^3\right)Q_\infty +E,
\end{equation}
with $E$ an effective divisor with support disjoint from $\{Q_0,Q_\infty\}.$ The above information is enough to determine that $H(Q_\infty),$ the Weierstrass semigroup of $Q_\infty$, is generated by $q^3, q\frac{q^e+1}{q+1}$ and $q^e+1$. 
\begin{theorem}[\cite{GOS}, Cor.3.5]\label{thm:GOS}
  We have $H(Q_\infty) = \displaystyle{\left\langle q^3, \, q\frac{q^e+1}{q+1}, \, q^e+1 \right\rangle}$.
\end{theorem}
%
A direct consequence of this theorem is a description of $\Gamma(Q_\infty)$, the set of gaps of $H(Q_\infty)$.
\begin{corollary}\label{cor:gammainf}
The set $\Gamma(Q_\infty)$ of gaps of $H(Q_\infty)$ is given by
\begin{align*}
  \biggl\{ k(q^e+1)+\ell & q\frac{q^e+1}{q+1}+mq^3 \mid \\
     & 0 \le \ell \le q, 0 \le m < \frac{q^e+1}{q+1}, k<0, k(q^e+1)+\ell q\frac{q^e+1}{q+1}+mq^3 \ge 0 \biggr\}.
\end{align*}
\end{corollary}
\begin{proof}
Any integer can uniquely be written in the form $k(q^e+1)+\ell q\frac{q^e+1}{q+1}+mq^3$, with $k,\ell$ and $m$ integers satisfying  $0 \le \ell \le q, 0 \le m < \frac{q^e+1}{q+1}$. To be an element of $H(Q_\infty)$ the additional requirement is simply that $k \ge 0$. Since $\Gamma(Q_\infty)=\mathbb{N}\setminus H(Q_\infty)$, the corollary follows.
\end{proof}

We now give a further consequence of Theorem \ref{thm:GOS}: a complete description of the ring of functions that are regular outside $Q_\infty$; that is to say, the functions having no poles except possibly at $Q_\infty$. The next result follows directly from the similar statement in \cite[Prop.\,3.4]{GOS}.
\begin{corollary}\label{cor:R1}
The ring $R(Q_\infty)$ of functions in $\F_{q^{2e}}(\chi_e)$ regular outside $Q_\infty$ is given by $\F_{q^{2e}}[x,y,z].$
\end{corollary}
For the AG codes that we wish to study, we in fact need to understand a larger ring of functions, allowing functions that may have a pole in $Q_\infty$ as well as $Q_0$. An explicit description of this ring is given in the following corollary.
\begin{corollary}\label{cor:R2}
The ring $R(Q_0,Q_\infty)$ of functions in $\F_{q^{2e}}(\chi_e)$ regular outside $\{Q_0,Q_\infty\}$ is given by $\F_{q^{2e}}[x,x^{-1},y,z].$
\end{corollary}
\begin{proof}
It is clear from \cref{eq:divxyz} that any function in $\F_{q^{2e}}[x,x^{-1},y,z]$ is regular outside $\{Q_0,Q_\infty\}$. Conversely, if a function $f$ has no pole outside $\{Q_0,Q_\infty\}$, then for a suitably chosen exponent $k$, the function $x^kf$ has no pole outside $Q_\infty$. Hence $x^kf \in R(Q_\infty)$. \Cref{cor:R1} implies that $f \in \F_{q^{2e}}[x,x^{-1},y,z].$
\end{proof}

\Cref{cor:R2} implies that the ring $R(Q_0,Q_\infty)$ has a natural module structure over $\F_{q^{2e}}[x,x^{-1}]$. When viewed as such a module, $R(Q_0,Q_\infty)$ is free of rank $q^e+1$ with basis $y^\ell z^m$ where $0 \le \ell <q+1$ and $0 \le m < \frac{q^e+1}{q+1}$. For $e=1$, the above theorem and the mentioned consequences are well known. For $e=3$, these results are contained in \cite{GK,D}.

We now turn to the study of the two-point Weierstrass semigroup $H(Q_0,Q_\infty).$ We will determine this semigroup completely. \Cref{eq:divxyz} will be used to describe the functions $f^{(i)}_{Q_0,Q_\infty}$, resp.~the bijection $\tau_{Q_0,Q_\infty}.$ For convenience, we will use the more compact notation $f^{(i)}_{0,\infty}$, resp.~$\tau_{0,\infty}$. Similarly we write $\tau_{0,\infty}^{-1}=\tau_{\infty,0}$.
%

\begin{theorem}\label{thm:fi}
Let $i \in \Z$ and write $i=-k(q^e+1)-\ell\frac{q^e+1}{q+1}-m$ for a triple $(k,\ell,m) \in \Z^3$ satisfying $0 \le \ell < q+1$ and $0 \le m < \frac{q^e+1}{q+1}$. Then $f_{0,\infty}^{(i)}=x^ky^\ell z^m.$
\end{theorem}
\begin{proof}
By definition of $f^{(i)}_{0,\infty}$ we have $-v_{Q_0}(f_i)=i$ and $v_{Q_\infty}(f_i)=\tau_{0,\infty}(i).$ Suppose $f^{(i)}_{0,\infty}$ cannot be chosen as a monomial in $x^{-1},x,y$ and $z$. Since by \cref{cor:R2} we have $f^{(i)}_{0,\infty} \in \F_{q^{2e}}[x,x^{-1},y,z]$ we can write $$f^{(i)}_{0,\infty}=\sum_{\alpha=-N}^M\sum_{\beta=0}^q \sum_{\gamma=0}^{\frac{q^e-q}{q+1}} a_{\alpha\beta\gamma}x^\alpha y^\beta z^\gamma,$$ for integers $N,M$ and constants $a_{k\ell m} \in \F_{q^{2e}}.$ Note that the pole orders at $Q_0$ of each of the occurring monomials $x^\alpha y^\beta z^\gamma$ are distinct. Since $-v_{Q_0}(f^{(i)}_{0,\infty})=i$, this implies that there exists a uniquely determined triple $(k,\ell,m)$ such that $a_{k \ell m} \neq 0$ and $i=-k(q^e+1)-\ell\frac{q^e+1}{q+1}-m$, while for all other monomials $x^\alpha y^\beta z^\gamma$ occurring in $f^{(i)}_{0,\infty}$ we have $$-\alpha(q^e+1)-\beta\frac{q^e+1}{q+1}-\gamma < i.$$

Likewise, the pole orders at $Q_\infty$ of all of the occurring monomials $x^\alpha y^\beta z^\gamma$ are distinct. Since $-v_{Q_\infty}(f^{(i)}_{0,\infty})=\tau_{0,\infty}(i)$ there exists a uniquely determined triple $(k',\ell',m')$ such that $a_{k' \ell' m'} \neq 0$ and $\tau_{0,\infty}(i)=k'(q^e+1)+\ell'q\frac{q^e+1}{q+1}+m'q^3$, while for all other monomials $x^\alpha y^\beta z^\gamma$ occurring in $f^{(i)}_{0,\infty}$ we have $$\alpha(q^e+1)+\beta q\frac{q^e+1}{q+1}+\gamma q^3< \tau_{0,\infty}(i).$$

If $(k,\ell,m) \neq (k',\ell',m')$, the monomial $x^k y^\ell z^m$ would have pole order $i$ in $Q_0$, but pole order strictly less than $\tau_{0,\infty}(i)$ in $Q_\infty$, which gives a contradiction by the definition of $\tau_{0,\infty}$. Hence we may take $f^{(i)}_{0,\infty}=x^k y^\ell z^m.$
\end{proof}

\begin{corollary}\label{cor:tau}
Let $i \in \Z$, and let $(k,\ell,m)\in \Z^3$ be the unique triple such that $0 \le \ell <q+1$, $0 \le m < \frac{q^e+1}{q+1}$ and $i=-k(q^e+1)-\ell\frac{q^e+1}{q+1}-m$. Then $$\tau_{0,\infty}(i)=k(q^e+1)+\ell q\frac{q^e+1}{q+1}+mq^3.$$
\end{corollary}
\begin{proof}
For a given $i \in \Z$, the proof of \cref{thm:fi} implies that $f^{(i)}_{0,\infty}=x^k y^\ell z^m$ for a uniquely determined triple $(k,\ell,m) \in \Z^3$ such that $-i=k(q^e+1)+\ell\frac{q^e+1}{q+1}+m$, $0 \le \ell \le q$ and $0 \le m < \frac{q^e+1}{q+1}.$ Hence $\tau_{0,\infty}(i)=-v_{Q_\infty}(f_i)=k(q^e+1)+\ell q\frac{q^e+1}{q+1}+mq^3$ as claimed.
\end{proof}

It is interesting to see what can be said about the Weierstrass semigroups $H(Q_0)$ and $H(Q_\infty)$ using the above tools. First of all, it should be noted that for $e=1$ and $e=3$, it is well known that $H(Q_0)=H(Q_\infty)$. The reason is that there exists an automorphism interchanging $Q_0$ to $Q_\infty$. For $e>3$, the place $Q_\infty$ is fixed by any automorphism of $\chi_e$ and in fact $H(Q_0)$ and $H(Q_\infty)$ were shown to be distinct in \cite{GOS}. However, for any $e \ge 1$ the points of the form $P_{(a,b,0)}$ fall within the same orbit under the action of the subgroup of the automorphism group of $\chi_e$ consisting of automorphisms fixing $Q_\infty$. This means that later on in the article, one can always exchange the point $Q_0$ with any point of the form $P_{(a,b,0)}$.

It is easy to describe the set $\Gamma(Q_0)$, but it should first be noted that the precise structure of $H(Q_0)$ (and hence of $\Gamma(Q_0)$) has already been determined in \cite{BMZ}. For the sake of completeness and since our description of $\Gamma(Q_0)$ is rather compact, we give the following corollary.
\begin{corollary}
The set $\Gamma(Q_0)$ of gaps of the Weierstrass semigroup $H(Q_0)$ of the point $Q_0$ on $\chi_e$ is given by
\begin{align*}
  \biggl\{-k(q^e+1)-\ell & \frac{q^e+1}{q+1}-m \mid \\
  & 0 \le \ell \le q, 0 \le m < \frac{q^e+1}{q+1}, k<0, k(q^e+1)+\ell q\frac{q^e+1}{q+1}+mq^3 \ge 0\biggr\}.
\end{align*}
\end{corollary}
\begin{proof}
We denote by $\Gamma(Q_\infty)$ (resp. $\Gamma(Q_0)$) the set of gaps of $Q_\infty$ (resp. $Q_0$). It is well known that $\tau_{\infty,0}$ gives rise to a bijection from $\Gamma(Q_\infty)$ to $\Gamma(Q_0)$. Since by Corollary \ref{cor:gammainf} we have
\begin{align*}
 \Gamma(Q_\infty)= \biggl\{ k(q^e+1)+\ell & q\frac{q^e+1}{q+1}+mq^3 \mid \\
     & 0 \le \ell \le q, 0 \le m < \frac{q^e+1}{q+1}, k<0, k(q^e+1)+\ell q\frac{q^e+1}{q+1}+mq^3 \ge 0 \biggr\},
\end{align*}
\cref{cor:tau} implies that $\Gamma(Q_0)$ is as stated.
\end{proof}

\section{Two-point AG codes on the generalized GK curve.}

Since the curves $\chi_e$ are maximal, they are good candidates to be used for the construction of error-correcting codes. Let the divisor $D$ be the sum of all the rational points of $\chi_e$ different from $Q_0$ and $Q_\infty$. If the support of a divisor $G$ consists of one rational point not in $\mathrm{supp}(D)$, the code $C_L(D,G)$ is called a one-point AG code. Similarly, if $G=a_1Q_0+a_2Q_\infty$, the code $C_L(D,G)$ is called a two-point code. By slight abuse of notation, the dual of a one-point code (resp.~two-point code) are sometimes also called one-point (resp.~two-point) codes, but we will only use the terminology for the codes $C_L(D,G)$. The main reason we do this is that for any divisor $G$ with $\mathrm{supp}(G) \cap \mathrm{supp}(D) = \emptyset$, there exists a divisor $H$ with $\mathrm{supp}(H) \cap \mathrm{supp}(D) = \emptyset$ such that $C_L(D,G)^\perp=C_L(D,H)$, but even if the support of $G$ is small, the support of $H$ might be large. Therefore, in our sense of the word, $C_L(D,G)^\perp=C_L(D,H)$ may not be a one-point or two-point code, even if $C_L(D,G)$ is.

Duals of one-point codes with defining divisor of the form $aQ_\infty$ or $aQ_0$ on the generalized GK curves were investigated in \cite{FG,BMZ}. As we will see below, their analysis of the parameters of these codes has direct implications for the study of the one-point codes $C_L(D+Q_0,aQ_\infty)$ and $C_L(D+Q_0,aQ_0)$ themselves. Duals of two-point AG codes on the GK curve (i.e. $e=3$) have been studied in \cite{CT}. As we will see, their analysis can be refined significantly, yielding more excellent AG codes. Furthermore, the case $e>3$ will be considered.

The theorem used in \cite{CT} (which comes from \cite[Thm.\,2.1]{M}) allows one to improve the Goppa bound by one for the minimum distance of a nontrivial code defined on an algebraic curve $\chi$ of the form $C_L(D,(a_1+b_1-1)Q_1+(a_2+b_2-1)Q_2)$, where $Q_1$ and $Q_2$ are rational points not in $\mathrm{supp}(D)$. Here, the four nonnegative integers $a_1,a_2,b_1,b_2$ should satisfy
\begin{enumerate}
  \item $a_1 \ge 1$,
  \item $L((a_1-1)Q_1+a_2Q_2)=L(a_1Q_1+a_2Q_2)$,
  \item $(b_1,b_2-1-t) \in \Gamma(Q_1;Q_2)$ for all $t$ satisfying $0 \le t \le \min\{b_2-1,2g-1-a_1-a_2\}$.
\end{enumerate}
In the next theorem we show that the order bound in the same situation improves upon the Goppa bound by at least one as well. Therefore, our results will automatically include all results in \cite{CT} as a special case. First note that $L((a_1-1)Q_1+a_2Q_2)=L(a_1Q_1+a_2Q_2)$ is equivalent to saying that $\tau_{Q_1,Q_2}(a_1) > a_2$ by \cref{thm:dimL}. Further the condition that $(b_1,b_2-1-t) \in \Gamma(Q_1;Q_2)$ for all $t$ satisfying $0 \le t \le \min\{b_2-1,2g-1-a_1-a_2\}$ is equivalent to the statement that $\tau_{Q_1,Q_2}(b_1) \ge b_2$ or $\tau_{Q_1,Q_2}(b_1) < b_2-1 -\min\{b_2-1,2g-1-a_1-a_2\}.$ With these reformulations in mind, we now show that \cref{prop:B} implies \cite[Thm.\,2.1]{M}.

\begin{theorem}
Let $a_1,a_2,b_1,b_2$ be nonnegative integers and write $G:=(a_1+b_1-1)Q_1+(a_2+b_2-1)Q_2$. Further suppose that $\tau_{Q_1,Q_2}(a_1)>a_2$.
\begin{enumerate}
\item If $\tau_{Q_1,Q_2}(b_1) \ge b_2$, then $\nu(Q_1;(b_2-1)Q_2,a_2Q_2)>\deg(G)-2g(\chi)+2$.
\item If $\tau_{Q_1,Q_2}(b_1) < b_2-1-\min\{b_2-1,2g-a_1-a_2\}$, then $\nu(Q_2;b_1Q_1,(a_1-1)Q_1)>\deg(G)-2g(\chi)+2$.
\end{enumerate}
In particular, in either case the minimum distance of the code $C_L(D,G)^\perp$ is at least $\deg(G)-2g+3.$
\end{theorem}
\begin{proof}
  If $\tau_{Q_1,Q_2}(b_1) \ge b_2$, then $a_1 \in \Gamma(Q_1;a_2Q_2)$ and $b_1 \in \Gamma(Q_1;(b_2-1)Q_2)$. Combining Remark \ref{rem:2.2} with (the proof of) \cref{lem:B} we see that $\nu(Q_1;(a_1-1)Q_1+a_2Q_2,b_1Q_1+(b_2-1)Q_2)>\deg(G)-2g(\chi)+2.$ Indeed, the term $q_{Q_1,(a_1-1)Q_1+a_2Q_2}(t)q_{Q_1,b_1Q_1+(b_2-1)Q_2}(t)$ will contribute to the coefficient of $t^{v_Q(G)+1}$ with at least $1$. From \cref{prop:B} and the Goppa bound applied to $C_L(D,G+Q_1)^\perp$, we see that $C_L(D,G)$ has minimum distance at least $\deg(G)-2g+3.$

If $\tau_{Q_1,Q_2}(b_1) < b_2$ and $\min\{b_2-1,2g-1-a_1-a_2\}=b_2-1$, then we have $(b_1,b_2-1-t) \in \Gamma(Q_1;Q_2)$ by assumption for all $t$ satisfying $0 \le t \le b_2-1$. This implies that $\tau_{Q_1,Q_2}(b_1) <0$. However, since $\tau_{Q_1,Q_2}(0)=0$ and $b_1 \ge 0$, we see that $(b_1,0) \in H(Q_1,Q_2)$, giving a contradiction. This situation can therefore not occur.

If $\tau_{Q_1,Q_2}(b_1) < b_2$ and $\min\{b_2-1,2g-1-a_1-a_2\}=2g-1-a_1-a_2$, then similarly as before we have $\tau(b_1) < b_2-2g+a_1+a_2.$ This implies that $b_2-1-t \in \Gamma(Q_2;b_1Q_1)$ for all $t$ satisfying $0 \le t \le 2g-1-a_1-a_2.$ On the other hand, we have $\tau_{Q_1,Q_2}(a_1) \in \Gamma(Q_2;(a_1-1)Q_1).$ Now using \cref{rem:tau}, note that $$b_2-2g-a_1-a_2 \le a_2+b_2-\tau_{Q_1,Q_2}(a_1) \le b_2-1.$$ Hence $a_2+b_2-\tau_{Q_1,Q_2}(a_1) \in \Gamma(Q_2;b_1Q_1).$ The term $q_{Q_2,(a_1-1)Q_1+a_2Q_2}(t)q_{Q_2,b_1Q_1+(b_2-1)Q_2}(t)$ will then contribute to the coefficient of $t^{v_Q(G)+1}$ with at least $1$. Hence $\nu(Q_2;(a_1-1)Q_1+a_2Q_2,b_1Q_1+(b_2-1)Q_2)>\deg(G)-2g(\chi)+2$.
\Cref{prop:B} and the Goppa bound applied to $C_L(D,G+Q_2)^\perp$, imply that $C_L(D,G)^\perp$ has minimum distance at least $\deg(G)-2g+3.$
\end{proof}

With the above theorem in place, we could in principle start to compute our lower bound on the minimum distance of the duals of two-point codes. Before doing that, we show in the remainder of this section that duals of two-point codes on the generalized GK curve are closely related to two-point codes. This means that our bounds not only can be applied to the duals of two-point codes, but to two-point codes themselves as well. In order to do this, we need to understand the structure of the rational point of $\chi_e$. The structure of these points is described explicitly in \cite{ABQ,GOS}. Since $e$ is odd, we write $e=2t+1$ for some nonnegative integer $t$. Apart from $Q_\infty$, all rational points are of the form $P_{(a , b, c)}$. There are $q^3$ rational places of the form $P_{(a , b , 0)}$ and $q^3(q^e+1)(q^{e-1}-1)$ of the form $P_{(a , b , c)}$ with $c \neq 0$. Both for $c=0$ and $c \neq 0$, the place $P_{(a, b , c)}$ is unramified in the degree $q^3$ extension $\F_{q^{2e}}(\chi_e)/\F_{q^{2e}}(z)$ by \cite{ABQ,GOS}. This means that there are exactly $(q^e+1)(q^{e-1}-1)$ possible nonzero values of $c \in \F_{q^{2e}}$ giving rise to $q^3$ rational places of $\F_{q^{2e}}(\chi_e)$ if the form $P_{(a , b , c)}$. By \cite{ABQ} these values of $c$ are exactly the roots of the polynomial
$$f:=1+\sum_{i=0}^{t-1} z^{\frac{q^e+1}{q+1}\left(q^{2i+2}-1+q^e-q\right)}+\sum_{i=0}^{t-1} z^{\frac{q^e+1}{q+1}\left( q(q^{2i+2}-1)\right)}.$$
Denoting, as before, by $D$ the divisor which is the sum of all rational points distinct from $Q_0$ and $Q_\infty$, this implies that
\begin{equation}\label{eq:divz}
(zf)=Q_0+D-q^3(q^{2e-1}-q^e+q^{e-1})Q_\infty.
\end{equation}
This expression is very useful to determine whether or not two two-point codes are equal. This comes in very handy, when computing the order bound using \cref{prop:B}, since one should only apply this proposition if the codes $C_L(D,G+Q)$ and $C_L(D,G)$ are distinct. We give a criterion in the following lemma.
\begin{lemma}\label{lem:dimandeq}
Let $\chi_e$ be the generalized GK curve over $\F_{q^{2e}}$ and let the divisor $D$ be the sum of all its rational places different from $Q_0$ and $Q_\infty$. Further let $G=a_1Q_0+a_2Q_\infty$ and $Q \in \{Q_0,Q_\infty\}$. Then
$$\dim(C_L(D,G))=\dim(L(G)) - \dim(L(G+Q_0-q^3(q^{2e-1}-q^e+q^{e-1})Q_\infty)).$$
Furthermore $C_L(D,G+Q)=C_L(D,G)$ if and only if \begin{multline*}\dim(L(G+Q))-\dim(L(G+Q+Q_0-q^3(q^{2e-1}-q^e+q^{e-1})Q_\infty))=\\ \dim(L(G))-\dim(L(G+Q_0-q^3(q^{2e-1}-q^e+q^{e-1})Q_\infty)).\end{multline*}
\end{lemma}
\begin{proof}
First, note that $\dim(C_L(D,G))=\dim(L(G))-\dim(L(G-D)).$ Since $\dim(L(G-D))=\dim(L(G-D+(zf)))$, the first part of the lemma follows from \cref{eq:divz}. Now applying this formula to compute the dimension of $\dim(C_L(D,G+Q)),$ the lemma follows.
\end{proof}
Since we know the map $\tau_{0,\infty}$ explicitly, it is very easy to check the above criterion using \cref{thm:dimL}.

The function $zf$ from equation \eqref{eq:divz} is also useful when identifying dual two-point codes and two-point codes. The standard way to identify the dual of an AG code $C_L(D,G)^\perp$ with a code of the form $C_L(D,H)$ is to identify a differential on the curve with simple poles in all evaluation points and residues in these points equal to $1$. Equation \eqref{eq:divz} implies that the differential $\omega:=\frac{1}{fz}dz$ has simple poles and nonzero residue in all rational points of the form $P_{(a , b , c)}$. More precisely, using the defining equations of the curve $\chi_e$ and equation \eqref{eq:divz}, we obtain that
\begin{equation}\label{eq:divomega}
(\omega)=-Q_0-D+(q^{2e+2}-q^{e+3}+2q^{e+2}-q^e+q^2-1)Q_\infty.
\end{equation}
Since the differential $\omega$ has simple poles in all the points in $D$, its residues at those points will all be nonzero. However, it turns out that in general these residues are not all $1$. Nonetheless, we can identify an explicit relation between the class of codes $C_L(D,G)$ and $C_L(D,G)^\perp$. Two codes $C_1$ and $C_2$ are called equivalent up to column multipliers, which we denote by $C_1 \cong C_2$, if there exist nonzero elements $a_1,\dots,a_n$ such that the map $\phi: \F_{q^e} \to \F_{q^e}$ defined by $\phi(v_1,\dots,v_n)=(a_1v_1,\dots,a_nv_n)$ satisfies $\phi(C_1)=C_2$. Note that the basic parameters of such codes $C_1$ and $C_2$, such as the minimum distance, are the same.

\begin{proposition}\label{prop:dual}
Let $\chi_e$ be the generalized GK curve over $\F_{q^{2e}}$ and let the divisor $D$ be the sum of all its rational places different from $Q_0$ and $Q_\infty$. Further let $G=a_1Q_0+a_2Q_\infty$. Then $C_L(D,G)^\perp \cong C_L(D,H)$, where $$H=-(a_1+1)Q_0+(q^{2e+2}-q^{e+3}+2q^{e+2}-q^e+q^2-1-a_2)Q_\infty.$$
\end{proposition}
\begin{proof}
Let $h:=(zf)'$ be the derivative of $zf$ with respect to the variable $z$. Then the differential $\eta=h \omega=(zf)'/(zf) dz$ has simple poles in all the rational points $P_{(a , b , c)}$ of $\chi_e$. Moreover, in each of those points, the residue of $\eta$ is equal to $1$. Therefore the standard theory of AG codes implies that $C_L(D,G)^\perp = C_L(D,H')$, with $H'=D-G+(\eta)=D-G+(h)+(\omega).$ Since $zf$ has simple roots only, its derivative $h$ is nonzero in $Q_0$ and the points in $D$. Hence the codes $C_L(D,H') \cong C_L(D,H)$, with $H=D-G+(\omega)$. Explicitly, the column multipliers are given by $(h(P))_{P \in \mathrm{supp}(D)}.$ Using \cref{eq:divomega}, the lemma follows.
\end{proof}

This proposition implies that the class of two-point codes $C_L(D,a_1Q_0+a_2Q_\infty)$ on the generalized GK curve is essentially the same as the class of codes of the form $C_L(D,a_1Q_0+a_2Q_\infty)^\perp$. In particular, the bounds on the minimum distance of codes of the form $C_L(D,a_1Q_0+a_2Q_\infty)^\perp$ will imply bounds for the minimum distance of codes of the form $C_L(D,a_1Q_0+a_2Q_\infty)$. Note that the above proof shows that $C_L(D,G)^\perp = C_L(D,H')$, where $$H'=-(a_1+1)Q_0+(q^{2e+2}-q^{e+3}+2q^{e+2}-q^e+q^2-1-a_2)Q_\infty+((zf)').$$ However, for our purposes this is less useful, since the divisor of $(zf)'$ may contains other points of $\chi_e$. Therefore $C_L(D,H')$ is in general not a two-point code, even if $C_L(D,G)$ is.

Using the same differential $\omega$ as above, we obtain the following corollary for one-point AG codes on the generalized GK curve.
\begin{corollary}
Let $\chi_e$ be the generalized GK curve over $\F_{q^{2e}}$ and let the divisor $D$ be the sum of all its rational places different from $Q_0$ and $Q_\infty$. Further let $G=aQ_\infty$. Then $C_L(D+Q_0,G)^\perp \cong C_L(D+Q_0,H)$, where $$H=(q^{2e+2}-q^{e+3}+2q^{e+2}-q^e+q^2-1-a)Q_\infty.$$
\end{corollary}

\section{Computation of the order bound and results}
\label{sec:computation}

Now that all the theoretical tools are in place, all that is left is to give the lower bounds that we obtain using the above theory as well as state the improvements on the MinT tables \cite{MinT}. We would also like to explain briefly how we computed these bounds. The given algorithm works for any of the generalized GK curves $\chi_e$. We have already seen that the explicit description of the bijection $\tau_{0,\infty}$ in \cref{cor:tau}, implies that for $G=a_1Q_0+a_2Q_\infty$ and $Q \in \{Q_0,Q_\infty\}$, it is computationally easy to determine:
\begin{enumerate}
\item The dimension of $L(G)$, see \cref{thm:dimL}.
\item The dimension of $C_L(D,G)$ (and hence of $C_L(D,G)^\perp$), see \cref{lem:dimandeq}.
\item The sets $H(Q;G)$ (and hence the value of $\nu(Q;G)$), see \cref{cor:Gnongaps}.
\end{enumerate}
What is left is to describe how to find the best recursive use of \cref{prop:B}. We do this efficiently by using a dynamic programming approach, in the form of a backtracking algorithm which starts with large degree divisors, where the order bound coincides with the Goppa bound and is easy to compute, and then backtracks to smaller degree divisors. A pseudo-code description is given in \cref{alg:order_bound}. For $q=2$ and $e=3$ our results supplement and improve those in \cite{CT}, as indicated in \cref{tab:min_dist}.

\begin{algorithm}[h!]
  \caption{\textbf{:} \textsc{OrderBoundTable}}
  \label{alg:order_bound}
  \newcommand{\obt}{\mathrm{orderBound}}
  \newcommand{\maxdeg}{\Delta}
  \newcommand{\degree}{\delta}
  \begin{algorithmic}[1]
    \Require{parameters $q$ and $e$.}
    \Ensure{array containing the order bound for $C_L(D,aQ_0+bQ_\infty)^\perp$, for $a,b\in\N$ whose sum $a+b$ is at most $\maxdeg$, a bound beyond which the order bound and the Goppa bound coincide.}
    \State $g_e \assign (q-1)(q^{e+1}+q^e-q^2)/2$ \eolcomment{genus of $\chi_e$}
    \State $N_e \assign q^{2e+2}-q^{e+3}+q^{e+2}+1$ \eolcomment{number of rational points}
    \State $\maxdeg \assign N_e + 2g_e$ \eolcomment{if larger degree, order bound coincides with Goppa bound}
    \State $\obt \assign $ two-dimensional array of size $(\maxdeg+1) \times (\maxdeg+1)$
    \For{$a$ \textbf{from} $\maxdeg$ \textbf{to} $0$}
      \State $\obt[a,\maxdeg-a] = \maxdeg - 2g_e + 2$ \eolcomment{Goppa bound for degree $\maxdeg$}
    \EndFor
    \For{$\degree$ \textbf{from} $\maxdeg-1$ \textbf{to} $0$} \eolcomment{backtrack: iterate on decreasing degree $\degree=a+b$}
    \For{$a$ \textbf{from} $0$ \textbf{to} $\degree$}
      \State $b \assign \degree-a$
      \State \inlcomment{Walk on the horizontal edge}
      \State $U \assign \{\text{Weierstrass semigroup at } Q_0\} \cap \{0,\ldots,\degree+1\}$
      \State $V \assign \{bQ_\infty\text{-non-gaps at }Q_0\} \cap \{-b, \ldots,a+1\}$
      \State $\overline{U} \assign \{a+1 - u, u \in U\}$
      \State $w \assign$ cardinality of $\overline{U} \cap V$
      \If{$w\neq 0$ and $\dim(C_L(D,aQ_0+bQ_\infty)) \neq \dim(C_L(D,(a+1)Q_0+bQ_\infty))$}
        \State $\mathrm{hbound} \assign \min(w,\obt[a+1,b])$
      \Else{ $\mathrm{hbound} \assign \obt[a+1,b]$}
      \EndIf
      \State \inlcomment{Walk on the vertical edge}
      \State $U \assign \{\text{Weierstrass semigroup at } Q_\infty\} \cap \{0,\ldots,\degree+1\}$
      \State $V \assign \{aQ_0\text{-non-gaps at }Q_\infty\} \cap \{-a, \ldots, b+1\}$
      \State $\overline{U} \assign \{b+1 - u, u \in U\}$
      \State $w \assign$ cardinality of $\overline{U} \cap V$
      \If{$w\neq 0$ and $\dim(C_L(D,aQ_0+bQ_\infty)) \neq \dim(C_L(D,aQ_0+(b+1)Q_\infty))$}
        \State $\mathrm{vbound} \assign \min(w,\obt[a,b+1])$
      \Else{ $\mathrm{vbound} \assign \obt[a,b+1]$}
      \EndIf
      \State \inlcomment{Combine the obtained bounds}
      \State $\obt[a,b] \assign \max(\mathrm{hbound},\mathrm{vbound})$
    \EndFor
    \EndFor
  \end{algorithmic}
\end{algorithm}

\begin{table}
  \centering
  \[
  \begin{array}{|c|c|c|c|c||c|c|c|c|c|}
  \hline
  n   & k   & (a_1,a_2) & d_{2P} & d_{1P} & n   & k   & (a_1,a_2) & d_{2P} & d_{1P} \\
  \hline
  223 & 222 & (0,0)     & 2      & 2  & 223 & 203 & (22,7)    &\mathbf{13}&12   \\
  223 & 221 & (6,0)     & 2      & 2  & 223 & 202 & (22,8)    & 13     & 12     \\
  223 & 220 & (8,0)     & 2      & 2  & 223 & 201 & (31,0)    & 14     & 14     \\
  223 & 219 & (11,0)    & 3      & 3  & 223 & 200 & (28,4)    &\mathbf{15}&14   \\
  223 & 218 & (13,0)    & 3      & 3  & 223 & 199 & (28,5)    & 16     & 15     \\
  223 & 217 & (14,0)    & 3      & 3  & 223 & 198 & (28,6)    & 17     & 16     \\
  223 & 216 & (8,7)     & 4      & 3  & 223 & 197 & (28,7)    &\mathbf{18}&17   \\
  223 & 215 & (16,0)    & 4      & 4  & 223 & 196 & (28,8)    &\mathbf{19}&18   \\
  223 & 214 & (17,0)    & 5      & 5  & 223 & 195 & (37,0)    & 20     & 20     \\
  223 & 213 & (19,0)    & 6      & 6  & 223 & 10  & (215,7)   & 205    & 204    \\
  223 & 212 & (20,0)    & 6      & 6  & 223 & 9   & (216,7)   & 206    & 206    \\
  223 & 211 & (21,0)    & 6      & 6  & 223 & 8   & (217,7)   & 207    & 206    \\
  223 & 210 & (22,0)    & 6      & 6  & 223 & 7   & (218,7)   & 208    & 207    \\
  223 & 209 & (19,4)    & 8      & 6  & 223 & 6   & (219,7)   & 209    & 208    \\
  223 & 208 & (19,5)    & 9      & 6  & 223 & 5   & (220,7)   & 211    & 209    \\
  223 & 207 & (19,6)    & 9      & 7  & 223 & 4   & (222,7)   & 214    & 212    \\
  223 & 206 & (19,7)    & 10     & 8  & 223 & 3   & (231,0)   & 215    & 215    \\
  223 & 205 & (19,8)    & 11     & 9  & 223 & 2   & (226,6)   & 217    & 214    \\
  223 & 204 & (28,0)    & 12     & 12 & 223 & 1   & (228,6)   & 223    & 220    \\
  \hline
  \end{array}
  \]
  \caption{\cref{tab:min_dist} gives for $q=2$, $e=3$, $n=223$ and fixed $k$ a value of $(a_1,a_2)$ for which the estimate $d_{2P}$ for the minimum distance of the code $C_L(D,a_1Q_0+a_2Q_\infty)^\perp$ is largest. It is compared to the corresponding estimate $d_{1P}$ for the minimum distance of a code of the same length and dimension of the form $C_L(D,a_1Q_0)^\perp$ or $C_L(D,a_2Q_\infty)^\perp$. The four entries in boldface indicate new improvements on the MinT tables. In \cite{CT} it was already shown that the entries for $k \in \{198,199\}$ improve the MinT \cite{MinT} table, which is why we have not put those two values in boldface.}
  \label{tab:min_dist}
\end{table}

%
%
%
%
%

\end{document}